\documentclass[letterpaper,11pt]{article}

\usepackage{geometry}
\usepackage[T1]{fontenc}  

\usepackage{graphicx} 
\usepackage{amsmath,amssymb,amsfonts}
\usepackage{enumerate}
\usepackage{hyperref}

\usepackage{caption}
\usepackage{subcaption}
\usepackage{wrapfig}
\usepackage{floatrow}
\newfloatcommand{capbtabbox}{table}[][\FBwidth]

\usepackage{color}
\usepackage[normalem]{ulem}

\usepackage{algorithm}
\usepackage{algpseudocode}
\floatname{algorithm}{Protocol}

\newcommand{\eps}{\varepsilon}
\newcommand{\lead}{{\tt leader}}
\newcommand{\foll}{{\tt follower}}
\newcommand{\formingjunta}{{\tt Forming\_junta}}
\newcommand{\polylog}{{\rm poly}\log}
\newcommand{\qed}{\hfill\ensuremath{\blacksquare\\}}%

\newtheorem{theorem}{Theorem}
\newtheorem{lemma}{Lemma}
\newtheorem{corollary}{Corollary}
\newtheorem{fact}{Fact}

\newenvironment{proof}{\noindent{\bf Proof:}}{\qed}

\date{}

\title{Fast Space Optimal Leader Election in Population Protocols\footnote{This work is sponsored in
part by the University of Liverpool initiative Networks Sciences and Technologies (NeST) and 
by the Polish National Science Centre grant DEC-2012/06/M/ST6/00459.}}

\author{{\bf Leszek G\k asieniec} \\ University of Liverpool \\ {\tt l.a.gasieniec@liverpool.ac.uk} \and {\bf Grzegorz Stachowiak} \\ Uniwersytet Wroc\l awski \\ {\tt gst@cs.uni.wroc.pl}}

\begin{document}

\maketitle

\thispagestyle{empty}

\begin{center}
{\bf Abstract}
\end{center}

The model of {\em population protocols} refers to the growing in popularity theoretical framework suitable for studying 
{\em pairwise interactions} within a large collection of simple indistinguishable entities, frequently called {\em agents}.
In this paper the emphasis is on the space complexity in fast {\em leader election} via population protocols governed by  
the {\em random scheduler}, which uniformly at random selects pairwise interactions from the population of $n$ agents.

The main result of this paper is a new fast and space optimal {\em leader election protocol}.
The new protocol operates in {\em parallel time} $O(\log^2 n)$
equivalent to $O(n\log^2 n)$ sequential {\em pairwise interactions}, 
in which each agent utilises $O(\log\log n)$ states. 
This double logarithmic space utilisation matches asymptotically the lower bound $\frac{1}{2}\log\log n$ on 
the number of states utilised by agents in any leader election algorithm with the running 
time $o(\frac{n}{{\rm polylog}\ n})$, see~\cite{AA+17}.

Our solution relies on the concept of phase clocks, 
a fundamental synchronisation and coordination tool in the field of Distributed Computing. 
We propose a new fast and robust population protocol for initialisation of phase clocks
to be run simultaneously in multiple modes and intertwined with the leader election process.
We also provide the reader with the relevant formal argumentation indicating that 
our solution is always correct and fast with high probability.   

%
%




\section{Introduction}\label{IN:sect}
The model of {\em population protocols} adopted in this paper was introduced in the seminal
paper of Angluin {\em et al.} \cite{AA+04}. Their model provides a universal theoretical framework for studying 
pairwise interactions within a large collection of 
indistinguishable entities, 
very often referred to as {\em agents} equipped with fairly limited communication and computation ability. 
The agents are modelled as finite state machines. When two agents engage in a direct interaction they mutually 
access the contents of their local states. On the conclusion of the encounter their states are modified 
according to the transition function that forms an integral part of the population protocol.
In the {\em probabilistic variant} of population protocols, considered in~\cite{AA+04} and adopted in this paper, 
in each step the {\em random scheduler} selects a pair of agents uniformly at random. 
In this variant in addition to the {\em space utilisation} determined by the maximum number of distinct states 
used by each agent, one is also interested in the {\em running time} of considered algorithmic solutions. 
More recent studies on population protocols focus on performance in terms of {\em parallel time} defined as the total number
of pairwise interactions leading to stabilisation divided by the size (in our case $n$) of the population. 
Please note that the parallel time can be also interpreted as the local time observed by agents.

A population protocol {\em terminates with success} if the whole population eventually stabilises, i.e., arrives at and 
stays indefinitely in the final configuration of states reflecting the desired property of the solution.
For example, in protocols targeting majority in the population, 
the final configuration corresponds to each agent being in the unique state representing 
the colour of the majority, see, e.g.,~\cite{AA+06,AA+08b,GH+15,GH+16,MN+14}.
In {\em leader election}, however, in the final configuration a single agent is expected to conclude 
in a $\lead$ state and all other agents must stabilise in $\foll$ states.
The leader election problem received in recent years greater attention in the context of population 
protocols thanks to several important developments in closely related problems~\cite{CC+14, D14}. 
In particular, the results from~\cite{CC+14, D14} laid down the foundation for the proof that leader election 
cannot be solved in sublinear time with agents utilising a fixed number of states~\cite{DS+15}.
In further work~\cite{AG+15}, Alistarh and Gelashvili studied the relevant upper bound, where they proposed 
a new leader election protocol stabilising in time $O(\log^3 n)$ assuming $O(\log^3 n)$ states at each agent.

In a very recent work Alistarh {\em et al.}~\cite{AA+17} consider a general trade-off between the number of states
used by agents and the time complexity of the stabilisation process.
In particular, the authors provide a separation argument distinguishing between {\em slowly stabilising} population protocols 
which utilise $o(\log\log n)$ states and {\em rapidly stabilising} protocols requiring $O(\log n)$ states at each~agent. 
This result nicely coincides with another fundamental observation by Chatzigiannakis {\em et al.}~\cite{CM+11} 
which states that population protocols utilising $o(\log\log n)$ states can cope only with semilinear predicates 
while the use of $O(\log n)$ states admits computation of symmetric predicates.\\

\paragraph{ Our results.}
In this paper we show that the space complexity lower bound in fast leader election proved in~\cite{AA+17} is asymptotically tight.
The lower bound states that any leader election algorithm with the time complexity $o(\frac{n}{{\rm polylog}\ n})$ 
requires $\frac{1}{2}\log\log n$ states in each agent.
In this paper we present a new fast {\em leader election} algorithm which stabilises in time $O(\log^2 n)$ 
in populations with agents utilising $c\log\log n$ states, for a small positive constant $c$.

In the most recent work on {\em majority problem} in population protocols~\cite{AAG17},
Alistarh {\em et al.} show a lower bound of $\Omega(\log n)$ states for any protocol
which stabilises in $O(n^c)$ time, for any constant $c\le 1.$ They also match this bound
from above by an algorithm which utilises $O(\log n)$ states at each agent, and stabilises 
in time $O(\log^2 n)$. 


Our algorithm utilises a fast and small space reduction of potential leaders (candidates) in the population.
The reduction process is intertwined with a robust initialisation and further utilisation of {\em phase clocks}, 
a core synchronisation tool developed and broadly explored in self-stabilising literature~\cite{H00}.  
This includes the seminal work on clock synchronisation by Arora {\em et al.}~\cite{AD+91},
further extension by Dolev and Welsh~\cite{DW04} to distributed systems prone to Byzantine faults, and 
related study on pulse synchronisation by Daliot {\em et al.}~\cite{DD+03}.
Our variant of the phase clock refers directly to the work of Angluin {\em et al.}~\cite{AA+08} in which the authors
propose efficient simulation of a {\em virtual register machine} supporting basic arithmetic operations.
The simulation in~\cite{AA+08} assumes availability of a single leader which coordinates the relevant exchange of information.
In the same paper, the authors provide also some intuition behind the phase clock coordinated by a {\em junta} 
of $n^\eps$ leaders, for a small positive constant $\eps.$
In this work we formally prove that the phase clock based on junta of cardinality $n^\eps,$ for any $\eps<1,$ allows
to count $\Theta(\log n)$ time units assuming a constant number of states at each agent.
We also consider an extension of the phase clock allowing to compute time $\Theta(\log^c n),$ for any integer constant $c.$
Our main result is based on rapid computation of junta of leaders followed by fast election of a single leader, 
all in time $O(\log^2 n)$ and $O(\log\log n)$ states available at each agent.   \\
%

\paragraph{Related work.} {\em Leader election} is one of the fundamental problems in the field of Distributed Computing on par with other core 
problems including {\em broadcasting}, {\em mutual-exclusion}, {\em consensus}, see an excellent text book by Attiya and Welch~\cite{AW04}.
The problem was originally studied in networks with nodes having distinct labels~\cite{L77}, where an
early work focuses on the ring topology in synchronous~\cite{FL87,HS80} as well as in asynchronous models~\cite{B80,P82}. 
Also, in networks populated by mobile agents the leader election was studied first in networks with labeled nodes~\cite{HK08}.
However, very often leader election is used as a powerful symmetry breaking mechanism enabling feasibility and coordination 
of more complex protocols in systems based on uniform (indistinguishable) entities.
There is a large volume of work~\cite{A80,AS91,AS+88,BS+96,BV99,YK89,YK96} on leader election in anonymous networks. 
In \cite{YK89,YK96} we find a characterisation of message-passing networks in which leader election is feasible when the nodes are anonymous. 
In~\cite{YK89}, the authors study the problem of leader election in general networks under the assumption that node labels
are not unique. In~\cite{FK+04}, the authors study feasibility and message complexity of leader election in rings with
possibly non-unique labels, while in~\cite{DP04} the authors provide solutions to a generalised leader election problem 
in rings with arbitrary labels. The work in~\cite{FP11} focuses on space requirements for leader election in unlabelled networks.
In~\cite{FP12}, the authors investigate the running time of leader election in anonymous networks where the time complexity
is expressed in terms of multiple network parameters. In~\cite{DP14}, the authors study feasibility of leader election
for anonymous agents that navigate in a network asynchronously.
Also, an interesting study on trade-offs between the time complexity and knowledge available in anonymous trees
can be found in recent work of Glacet {\em et al.}~\cite{GM+16}.

Finally, a good example of recent extensive studies on the exact space complexity in related models refers to plurality consensus.
In particular, in \cite{BF+16} Berenbrink {\em et al.} proposed a plurality consensus protocol for $C$ original opinions converging in 
$O(\log C\cdot\log\log n)$ synchronous rounds using only $\log C+ O(\log\log C)$ bits of local memory. 
They also show a slightly slower solution converging in $O(\log n\cdot\log\log n)$ rounds using only 
$\log C +4$ bits of local memory. This disproved the conjecture by Becchetti et al. \cite{BC+15} implying 
that any protocol with local memory $\log C +O(1)$ has the worst-case running time 
$\Omega(k).$
In \cite{GP16} Ghaffari and Parter propose an alternative algorithm converging in time $O(\log C \log n)$
in wich messages and local memory utilise $\log C + O(1)$ bits.
%
In addition, some work on the application of the random walk in plurality consensus protocols 
can be found in \cite{BC+15,GH+15}.


\section{Preliminaries}\label{PR:sect}
We consider population protocols defined on the complete graph of interactions where
the {\em random scheduler} picks uniformly at random pairs of agents drawn from the population of size $n$.
The agents are anonymous, i.e., they don't have identifiers. 
The protocol assumes all agents start in the same initial state.
Our protocol utilises the classical model of population protocols~\cite{AA+04,AA+08} where the consecutive interactions 
refer to ordered pairs of agents $({\tt responder},{\tt initiator})$.
On the conclusion of each interaction the two participating agents change their states $(a,b)$
into $(a',b')$ according to a {\em fixed deterministic transition function} denoted by $(a,b)\rightarrow (a',b').$

\paragraph{Random coins.} 
For the simplicity of presentation, in this paper we dispense fair random coins whp by observing actions of the random scheduler.
It has been shown, however, that agents can generate {\em  synthetic coins} which become almost uniform after a constant number of interactions~\cite{AA+17}.
This method is based on concentration properties of random walks on the hypercube, see, e.g.,~\cite{AR16}. 
A similar approach can be found in~\cite{CKL16} where Cardelli {\em et al.} generate randomness in chemical reaction networks.

We focus here on two complexity measures: (1) the {\em space complexity} defined as the {\em number of states} 
required by each agent, and (2) the {\em time complexity} reflecting the number of interactions needed to 
stabilise the population protocol.
Similarly to other recent work in the field, the emphasis in this paper is on {\em parallel time} of the solution 
defined as the total number of interactions divided by the size of the population.
This time can be also seen as the local time observed by an agent, i.e.,
the number of pairwise interactions in which the agent is involved in.  
In this work we aim at protocols formed of $O(n\cdot \polylog n)$ interactions
equivalent to the parallel running time $O(\polylog n).$

Our leader election algorithm is always correct and 
it runs fast {\em with high probability (whp)} which we define as follows.
Let $\eta$ be a universal constant referring to the quality of our protocols. We say that
an event occurs with {\em negligible} probability if it occurs with probability at most $n^{-\eta}$, and
an event occurs with high probability (whp) if it occurs with probability at least $1-n^{-\eta}$.
This estimate is of asymptotic nature, i.e., we assume $n$ is large enough to validate the results.
Similarly, we say that an algorithm succeeds with high probability if 
it succeeds with probability at least $1-n^{-\eta}$.
When we refer to a probability of failure $p$ different from $n^{-\eta}$, 
we say explicitly whp $1-p$.


\subsection{One-way epidemics}

In our solution we adopt the notion of {\em one-way epidemic} introduced in~\cite{AA+08}.
One-way epidemic refers to the population protocol with the state space $\{0,1\}$ and the transition rule $x,y\rightarrow \max\{x,y\},y$.
One interprets 0's as {\em susceptible} agents and 1's as {\em infected} ones. 
This protocol corresponds to a simple epidemic in which transmission of the infection occurs if and only if 
the initiator is infected and the responder is susceptible. 
We will use the following theorem introduced in~ \cite{AA+08}.

\begin{theorem}[\cite{AA+08}]\label{PR:oneway}
In order to conclude one-way epidemic (infect all agents) one needs 
$\Theta(n\log n)$ pairwise interactions with high probability.
\end{theorem}

\section{Phase clock revisited}\label{PC:sect}
In~\cite{AA+08} Angluin {\em et al.} defined and further analysed the concept of
{\em phase clocks} capable of counting parallel time $\Theta(\log n)$ approximately, 
in which each agent participating in the population protocol utilises a
constant number of states.
%
The phase clocks studied in \cite{AA+08} work under the assumption of having already determined a unique leader in the population. 
In the same paper, the authors argue without giving a formal proof that phase clocks should also work when 
the unique leader is replaced by a {\em junta} of $n^\eps$ leaders, 
for some unspecified constant $\eps.$
Further on, the authors suggest also 
that once the
phase clock is in motion the leadership team can be reduced to a single leader with the
help of coin tossing combined with propagation via one-way epidemic. 
This would allow election of a single leader in expected $\Theta(n\log^2 n)$ interactions 
if a junta phase clock is established.
%
In this paper we adopt a similar mechanism to determine a single leader, 
however here the junta team has to be computed first.
We implement this process in two loops, one nested inside the other.
The {\em internal loop} operates in (parallel) time $\Theta(\log n)$ equivalent to
$\Theta(n\log n)$ interactions allowing to distribute 1's via one-way epidemic.
This internal loop in principle mimics actions of a finite state phase clock.
The {\em external loop} is used to count $\Theta(\log n)$ executions of the internal loop.
The {\em external loop} is controlled by a finite state phase clock too.
However this time the agents execute clock operations more seldom, i.e., only when 
they act as responders for the first time after each full execution of the internal loop.
This way a single execution of the external phase clock refers to time $\Theta(\log n)$
counted by the internal loop, which is in turn equivalent to the total time $\Theta(\log^2 n).$

In this section we propose and analyse a modified version of phase clocks capable 
of counting approximately time $\Theta(\log n),$ under assumption that each agent
utilises a constant number of states and the junta of leaders is 
of cardinality $n^{1-\eps}$, for any constant $\eps:0<\eps<1$.
Without loss of generality and for some technical reasons 
we assume $\eps=\eps(k)=\frac{3}{3k+1},$ for a positive integer $k$.

The states of agents controlling the phase clock protocol are structured in pairs $(x,b)$.
The entry $b$ has value $\lead$ for leaders in the junta and $\foll$ for all other agents.
The entry $x$ represents a {\em phase} denoted by a number of an agent drawn 
from the set ${\mathbb Z}_m=\{0,1,2,\ldots,m-1\}$, for some integer constant $m$.
The phases can be interpreted as hours on the dial of an analogue clock.
The increment of clock phases is periodic and computed using the arithmetic modulo $m$ 
denoted by $+_m$.
We also define the maximum of two phases $x,y$ in set ${\mathbb Z}_m$ as:

\[ {\max}_m\{x,y\}=\left\{
   \begin{array}{c}
     \max\{x,y\}\mbox{ if } |x-y|\leq m/2\\
     \min\{x,y\}\mbox{ if } |x-y|>m/2\\
   \end{array} \right.\
\]
Finally we define the {\em circular order} (which is not partial) 
on ${\mathbb Z}_m$ as $x\leq_m y$ iff ${\max}_m\{x,y\}=y.$

Now we are ready to formally define the transition function 
in our version of phase clocks as\\

\[
(x,\foll),(y,b)\rightarrow
({\max}_m\{x,y\},\foll),(y,b),
\]
and\\

\[
(x,\lead),(y,b)\rightarrow
({\max}_m\{x,y+_m 1\},\lead),(y,b).
\]

In this paper we study phase clocks which operate in two (nested in one another) modes:
the {\em ordinary mode} (analogue of the internal loop) and the {\em external mode} (analogue of the external loop).
In each mode, we say the phase clock 
{\em passes through $0$} whenever its current phase $x$ is reduced in 
absolute terms (e.g.,\ changes from phase 5 to phase 3).
As hinted earlier the two modes differ in selecting pairwise 
interactions to the relevant phase clock actions.
In particular:

\begin{itemize}
\item
In the {\bf ordinary mode} all interactions triggered by the random
scheduler prompt actions of the phase clock.
And once the ordinary mode clock passes through zero a {\em meaningful interaction}
of external mode occurs, i.e.,

\item
A {\bf meaningful interaction} refers to the first interaction of an agent after its ordinary phase clock 
passes through $0$ in which also the agent acts as the responder.

\item In the {\bf external mode}, however, interactions used 
to propel actions of the phase clock form {\em series} 
of $n$ interactions in which every agent acts as the responder exactly once. 
In each subsequent series the initiators are chosen at random (by the random scheduler)
and the order in which agents appear as responders is irrelevant.
In our algorithms such series are series of meaningful interactions
following passes through zero of the ordinary mode clock.
\end{itemize}

Before we proceed with the full proof of Theorem \ref{PC:main}, i.e.,
the main result of this section, we share with the reader several useful lemmas.
In the proofs referring to the ordinary mode we utilise Theorem \ref{PR:oneway} 
showing that one-way epidemic protocol concludes after $\Theta(n\log n)$ interactions whp. 
And for the external mode we need an analogue of this theorem.

\begin{lemma}\label{PC:1wayext}
One-way epidemic applied in the external mode requires $O(n\log n)$ interactions whp to stabilise.
\end{lemma}

\begin{proof}
Let $v$ be the first infected agent.
By the Chernoff bound, for any constant $c_1>0$ the number of 
interactions agent $v$ needs to infect directly $c_1\log n$ agents is bounded 
by $O(n\log n)$ whp $1-n^{-\eta-1}$.
Thus the number of infected agents after $O(n\log n)$ interactions is 
at least $c_1\log n$ whp $1-n^{-\eta-1}$.
Also by the Chernoff bound, there exists a constant $c_2>0,$ s.t., 
if the number of infected agents before a {\em series} of $n$ interactions is $A,$ 
where $c_2\log n<A<n/2$, then on the conclusion of the series the number of 
infected agents is at least $\frac{5}{4}\cdot A$ whp $1-n^{-\eta-1}$.
Thus if we take $c_1\geq c_2$, thanks to the exponential growth, 
the number of infected agents after $O(n\log n)$ interactions is 
at least $n/2$ whp $1-O(n^{-\eta-1}\log n)$.
Furthermore, by taking an extra $c_3 n\log n$ pairwise interactions 
each uninfected (yet) agent interacts $c_3\log n$ times as the responder.
One can choose a constant $c_3,$ s.t., the probability of not getting infected 
during these interactions is at most $n^{-\eta-1}$ for a fixed uninfected agent.
Finally, by the Union bound the probability of failure in any of these steps is
at most $n^{-\eta-1}+O(n^{-\eta-1}\log n)+n^{-\eta-1}\cdot n/2<n^{-\eta}$.
\end{proof}

For the simplicity of presentation we assume in the next few lemmas that 
the agents start in phase $0$.
The main purpose of these lemmas is to bound from above the sizes of sets 
of agents in phases $1,2,3,\ldots$ on the conclusion of $O(n\log n)$ interactions. 
There are separate collections of lemmas for the ordinary and the external modes.
Also here we assume $\eps=\frac{3}{3k+1}$ and $k<m/4$.
In what follows we state two lemmas with similar proofs for phase clocks in 
the ordinary and the external modes.

\begin{lemma}\label{PC:techord}
Assume $j\leq k$ and interactions of the phase clock are performed in the ordinary mode.
Assume also that at some point the number of agents in phase $x\geq_m i$
is at most $A\cdot n^{1-i\eps},$ for all $i=0,1,\ldots,j$ and some value $A\in[1,n^{\eps/3}]$.
Then after $n/4$ interactions the number of agents in phase $x\geq_m i$ is at most
$3A\cdot n^{1-i\eps},$ for all $i=0,1,\ldots,j$ and whp $1-2jn^{-10}$.
\end{lemma}

\begin{proof}
We prove this lemma by induction on~$j$.
For $j=0$ the thesis holds since the number of agents in phase $x\geq_m 0$
is at most $n<3A\cdot n^{1-0\cdot\eps}$ with probability~$1$.
Assume now, the thesis is true for $j-1$ and we prove it for $j$.
By the inductive assumption after the series of interactions 
the number of agents in phase $x\geq_m i$ is bounded from above 
by $3A\cdot n^{1-i\eps},$ for all $i=0,1,\ldots,j-1$ whp $1-2(j-1)n^{-10}$.
Two types of agents can enter phase $x\geq_m j$ during these interactions.

The first type refers to the leaders.
A leader can enter phase $x\geq_m j$ if it acts as the responder 
in the interaction with some initiator in phase $y\geq_m j-1.$
Assume the number of the potential initiators is at most $3A\cdot n^{1-(j-1)\eps}$,
which happens according to the inductive hypothesis whp.
Thus the probability $p_{\iota}$ that a new leader enters phase $x\geq_m j$
during any of $n/4$ interactions $\iota$
is at most $3A\cdot n^{1-(j-1)\eps}n^{1-\eps}/n^2=3A\cdot n^{-j\eps}$.
We attribute a binary $0$-$1$ sequence $\sigma$ of length $n/4$ to these interactions.
Initially $\sigma$ is empty and during each interaction $\iota$ we 
pad $\sigma$ with one bit as follows.
If a new leader in phase $x\geq_m j$ appears, we add $1$ to $\sigma$.
If no new leader in phase $x\geq_m j$ is selected, $1$ is inserted to $\sigma$ but 
only with probability $(3A\cdot n^{-j\eps}-p_{\iota})/(1-p_{\iota})$ and $0$ otherwise.
This way all entries of $\sigma$ are independently equal to $1$ 
with probability $3An^{-j\eps}$.
If the number of $1$s in $\sigma$ is smaller or equal to $A\cdot n^{1-j\eps}$,
the number of new leaders in phase $x\geq_m j$ is not larger than $A\cdot n^{1-j\eps}$.
The expected number of $1$s in 
$\sigma$ is $\frac{3}{4}A\cdot n^{1-j\eps}\geq \frac{3}{4}n^{\eps/3}$.
By the Chernoff bound, the probability this number is larger than $A\cdot n^{1-j\eps}$ is
negligible and smaller than $e^{-n^{\eps/3}/36}<n^{-10},$ for sufficiently large~$n$.
Thus the number of new leaders in phase $x\geq_m j$ is not larger than $A\cdot n^{1-j\eps}$ whp $1-n^{-10}$.

The second type of new agents in phase $x\geq_m j$ refers to followers.
A follower enters phase $x\geq_m j$, if it is a responder to an initiator
in phase $y\geq_m j$.
Also here we attribute a $0$-$1$ sequence $\rho$ of length $n/4$ to the relevant interactions.
Prior to these $n/4$ interactions $\rho$ is empty.
During each interaction $\iota$ gets extended $\rho$ by a single bit.
Let $p_{\iota}$ be the probability of getting a new follower in 
phase $x\geq_m j$ in a subsequent interaction $\iota$.
If $p_{\iota}>3A\cdot n^{-j\eps}$, then $1$ is inserted to $\rho$ 
with probability $3A\cdot n^{-j\eps}$ and $0$ otherwise.
If $p_{\iota}\leq 3A\cdot n^{-j\eps}$ and a new follower in phase 
$x\geq_m j$ appears, $1$ is added to $\rho$.
If $p_{\iota}\leq 3A\cdot n^{-j\eps}$ and no new follower in phase 
$x\geq_m j$ appears, then $1$ is added to $\rho$ with probability 
$(3An^{-j\eps}-p_{\iota})/(1-p_{\iota})$ and $0$ otherwise.
Note that until more than $A\cdot n^{1-j\eps}$ new followers occur 
in phase $x\geq_m j$, $p_{\iota}\leq 3A\cdot n^{1-j\eps}/n = 3A\cdot n^{-j\eps}$.
If the number of $1$s in $\rho$ is smaller or equal to $A\cdot n^{1-j\eps}$,
the number of new followers in phase $x\geq_m j$ is not larger 
than $A\cdot n^{1-j\eps}$.
The expected number of $1$s in $\rho$ is $\frac{3}{4}A\cdot n^{1-j\eps}\geq \frac{3}{4}n^{\eps/3}$.
By the Chernoff bound the probability that this number is larger than $A\cdot n^{1-j\eps}$ 
is negligible and smaller than $e^{-n^{\eps/3}/36}<n^{-10},$ for sufficiently large $n$.
Thus the number of new followers in phase $x\geq_m j$ is not larger than 
$A\cdot n^{1-j\eps}$ whp at least $1-n^{-10}$.
This concludes the proof that the number of agents in phase $x\geq_m i$ is at most 
$3A\cdot n^{1-i\eps},$ for all $i=0,1,\ldots,j$ whp $1-2jn^{-10}$.
\end{proof}

We formulate now the analogous lemma for the external mode.

\begin{lemma}\label{PC:techext}
Assume $j\leq k$ and interactions of the phase clock are performed in the external mode.
Assume also that at some point the number of agents in phase $x\geq_m i$
is at most $A\cdot n^{1-i\eps},$ for all $i=0,1,\ldots,j$ and some value $A\in[1,n^{\eps/3}]$.
Consider a series of at most $n/4$ interactions in which at most $n^{1-i\eps}/4$ leaders act as responders.
After this series of interactions, the number of agents in phase $x\geq_m i$ is at most
$3A\cdot n^{1-i\eps},$ for all $i=0,1,\ldots,j$ and whp at least $1-2jn^{-10}$.
\end{lemma}

\begin{proof}
We prove the lemma by induction on~$j$.
For $j=0$ the thesis holds since the number of agents in phase $x\geq_m 0$
is at most $n<3A\cdot n^{1-0\cdot\eps}$ with probability~$1$.
Assume now, the thesis is true for $j-1$ and we prove it for $j$.
By the inductive assumption after the series of interactions from the Lemma's thesis
the number of agents in phase $x\geq_m i$ is bounded from above 
by $3A\cdot n^{1-i\eps},$ for all $i=0,1,\ldots,j-1$ whp $1-2(j-1)n^{-10}$.
Two types of agents can enter phase $x\geq_m j$ during these interactions.

The first type refers to the leaders.
A leader can enter phase $x\geq_m j$ if it acts as the responder 
in the interaction with some initiator in phase $y\geq_m j-1.$
Assume the number of the potential initiators is at most $3A\cdot n^{1-(j-1)\eps}$,
which happens according to the inductive hypothesis whp.
There are at most $n^{1-\eps}/4$ interactions $\iota$ in the series in
which a leader is the responder.
The probability $p_{\iota}$ that such a leader enters phase $x\geq_m j$
during interaction $\iota$ is  at most $3A\cdot n^{-(j-1)\eps}$.
We attribute a binary $0$-$1$ sequence $\sigma$ of length $n^{1-\eps}/4$ to these interactions.
Initially $\sigma$ is empty and during each interaction $\iota$ we 
pad $\sigma$ with one bit as follows.
If a new leader in phase $x\geq_m j$ appears, we add $1$ to $\sigma$.
If no new leader in phase $x\geq_m j$ is selected, $1$ is inserted to $\sigma$ but 
only with probability $(3A\cdot n^{-(j-1)\eps}-p_{\iota})/(1-p_{\iota})$ and $0$ otherwise.
This way all entries of $\sigma$ are independently equal to $1$ 
with probability $3A\cdot n^{-(j-1)\eps}$.
If $\sigma$ has less than $n^{1-\eps}/4$ entries we add lacking entries
by independent coin tosses each time obtaining 1 with probability $3A\cdot n^{-(j-1)\eps},$
and $0$ with the remaining probability.
If the number of $1$s in $\sigma$ is smaller or equal to $A\cdot n^{1-j\eps}$,
the number of new leaders in phase $x\geq_m j$ is not larger than $A\cdot n^{1-j\eps}$.
The expected number of $1$s in 
$\sigma$ is $\frac{3}{4}A\cdot n^{1-j\eps}\geq \frac{3}{4}n^{\eps/3}$.
By the Chernoff bound, the probability this number is larger than $A\cdot n^{1-j\eps}$ is
negligible and smaller than $e^{-n^{\eps/3}/36}<n^{-10},$ for sufficiently large~$n$.
Thus the number of new leaders in phase $x\geq_m j$ is not larger than $A\cdot n^{1-j\eps}$ whp $1-n^{-10}$.

The second type of new agents in phase $x\geq_m j$ refers to followers.
A follower enters phase $x\geq_m j$, if it is a responder to an initiator
in phase $y\geq_m j$.
Also here we attribute a $0$-$1$ sequence $\rho$ of length $n/4$ to the relevant interactions.
Prior to these at most $n/4$ interactions $\rho$ is empty.
During each interaction $\iota$ gets extended $\rho$ by a single bit.
Let $p_{\iota}$ be the probability of getting a new follower in 
phase $x\geq_m j$ in a subsequent interaction $\iota$.
If $p_{\iota}>3A\cdot n^{-j\eps}$, then $1$ is inserted to $\rho$ 
with probability $3A\cdot n^{-j\eps}$ and $0$ otherwise.
If $p_{\iota}\leq 3A\cdot n^{-j\eps}$ and a new follower in phase 
$x\geq_m j$ appears, $1$ is added to $\rho$.
If $p_{\iota}\leq 3A\cdot n^{-j\eps}$ and no new follower in phase 
$x\geq_m j$ appears, then $1$ is added to $\rho$ with probability 
$(3An^{-j\eps}-p_{\iota})/(1-p_{\iota})$ and $0$ otherwise.
Note that until more than $A\cdot n^{1-j\eps}$ new followers occur 
in phase $x\geq_m j$, $p_{\iota}\leq 3A\cdot n^{1-j\eps}/n = 3A\cdot n^{-j\eps}$.
If the number of $1$s in $\rho$ is smaller or equal to $A\cdot n^{1-j\eps}$,
the number of new followers in phase $x\geq_m j$ is not larger 
than $A\cdot n^{1-j\eps}$.
The expected number of $1$s in $\rho$ is $\frac{3}{4}A\cdot n^{1-j\eps}\geq \frac{3}{4}n^{\eps/3}$.
By the Chernoff bound the probability that this number is larger than $A\cdot n^{1-j\eps}$ 
is negligible and smaller than $e^{-n^{\eps/3}/36}<n^{-10},$ for sufficiently large $n$.
Thus the number of new followers in phase $x\geq_m j$ is not larger than 
$A\cdot n^{1-j\eps}$ whp at least $1-n^{-10}$.
This concludes the proof that the number of agents in phase $x\geq_m i$ is at most 
$3A\cdot n^{1-i\eps},$ for all $i=0,1,\ldots,j$ whp at least $1-2jn^{-10}$.
\end{proof}

All lemmas below apply to both (the ordinary and the external) modes of the phase clock.

\begin{lemma}\label{PC:boundk}
Assume all agents start in the clock phase 0.
The probability that after $\frac{1}{8(3k+1)} n\log_3 n$ 
interactions (in either of the phase clock modes)
there are at least $n^{2/(3k+1)}$ agents in phases $x\geq_m k$
is at most $2(\eps/3)k\log_3 n\cdot n^{-10}$.
\end{lemma}

\begin{proof}
In the beginning the number of agents in phase $0$ is $n$ and there are no agents in any other phase.
So the number of agents in phases $x\geq_m i$
is at most $3A\cdot n^{1-i\eps},$ for all $i=0,1,\ldots,k$ and $A=1$.
To conclude the proof we apply Lemma \ref{PC:techord} (or Lemma \ref{PC:techext}, respectively to the mode)
$\frac{1}{3k+1}\log_3 n$ times for the series of $\frac{1}{8(3k+1)}n\log_3 n$ subsequent interactions.
For the ordinary mode we get the thesis by applying Lemma \ref{PC:techord} to subsequent series of $n/4$ interactions.

For the extended mode we can split each series of $n$ interactions into eight subseries.
In each subseries there are at most $n/4$ interactions and in at most $n^{1-\eps}$ interactions leaders acts as responders.
We can apply Lemma \ref{PC:techext} to these subseries to obtain the thesis.

Namely $A$ in these applications is equal $1,3,9,\ldots,n^{1/(3k+1)}/3=n^{\eps/3}/3$, and by Lemmas \ref{PC:techord} and \ref{PC:techext}
the number of agents in phase $x\geq_m k$ after all $\frac{1}{8(3k+1)} n\log_3 n$
interactions exceeds $n^{\eps/3}n^{1-k\eps}=n^{2/(3k+1)}$ with probability at most
$2(\eps/3)k\log_3 n\cdot n^{-10}$.
\end{proof}

\begin{lemma}\label{PC:series}
Assume all agents start in the clock phase 0.
The probability that on the conclusion of $\frac{n\log_3 n}{8(3k+1)}$ interactions
(in either of the phase clock modes)
there are some agents in clock phase $x\geq_m k+1$
is $O(n^{-\eps/3}\log n)$.
\end{lemma}

\begin{proof}
The (clock) phase $x=k+1$ can be entered only by a leader which
acts as the responder in the interaction with an agent in clock phase $x=k$.
Since the number of agents in clock phase $x=k$ is at most $n^{2\eps/3}$,
the probability of having such interaction in each series of $n$
interactions is at most $nn^{1-\eps}n^{2\eps/3}/n^2=n^{-\eps/3}$.
By the Union bound the probability of having such interaction during
$\frac{n\log_3 n}{8(3k+1)}$ subsequent interactions is $O(n^{-\eps/3}\log n)$.
\end{proof}

\begin{lemma}\label{PC:progress}
Assume all agents start in clock phase $x=0$ and $d$ is a positive constant.
Then there exists an integer constant $K<m/2,$ s.t., the first agent enters 
phase $x=K$ before interaction $t+dn\log n$ with negligible probability for sufficiently large $n.$
\end{lemma}

\begin{proof}
Assume $K=\kappa k$.
We can divide all phases $x=1,2,\ldots K$ into
$\kappa$ consecutive chunks having $k$ phases each.
Let $t_i,$ for all $i=0,1,2,\ldots,\kappa-1,$ be the first interaction
in which an agent enters phase $i\cdot k+1,$ where $t_0=0$.
Note that just before interaction $t_i$ all agents are in phases $x\leq_m ik$.
Thus after each subsequent interaction all agent phases are not larger ($\leq_m$)
as if they all started from phase $i\cdot k$ just before interaction $t_i$.
By Lemma~\ref{PC:series} the probability that $t_i-t_{i-1}<\frac{n\log_3 n}{8(3k+1)}$
is smaller than $cn^{-\eps/3}\log_3 n,$ for some constant $c>0$.
The probability, that for at least $\kappa'$ different values $i$ we have 
$t_i-t_{i-1}\leq\frac{n\log_3 n}{8(3k+1)}$
is by Union bound smaller than
\[{\kappa\choose \kappa'} \left(cn^{-\eps/3}\log_3 n\right)^{\kappa'}.\]
Now take $\kappa'>3\eta/\eps$ and $\kappa-\kappa'>d\cdot 8(3k+1)$.
Thus for sufficiently large $n$ we obtain $t_{\kappa}\leq dn\log n$ with probability at most $n^{-\eta}$.
\end{proof}

\begin{lemma}\label{PC:move}
For any constant $d$ there is another constant $K,$ s.t.,
if $m>6K$ and after interaction $t$ there is an agent in phase $i$ 
and all other agents are in phases $x:i-2K\leq_m x\leq_m i$, then whp
\begin{itemize}
\item the first interaction $t'$ when an agent enters phase $i+K$ 
satisfies $t'>t+dn\log n$, and
\item during interaction $t'$ all agents are in phases $x,$ s.t., $i\leq_m x\leq_m i+K$.
\end{itemize}
\end{lemma}

\begin{proof}
By Theorem \ref{PR:oneway} and Lemma \ref{PC:1wayext} there exists 
a positive constant $d',$ s.t., 
one-way epidemic succeeds within $d'\cdot n\log n$ interactions whp.
On the other hand by Lemma~\ref{PC:progress}, for a constant $D=\max\{d,d'\}$ there is $K,$ s.t.,
all agents starting in phase $i$ move to phase smaller or equal $i+K-1$ after 
$D n\log n$ interactions whp.
It is easy to see that the same holds if all agents start in phases $x:i-2K\leq_m x\leq_m i$.
Thus $t'>Dn\log n\geq dn\log n$ whp.
Since one way epidemic initiated by an agent in phase $i$ during interaction 
$t$ succeeds whp, all agents after interaction $t'$ are in phase $x\geq_m i$ whp.
\end{proof}

Consider now the interactions in which phase clocks in different agents pass through $0.$
We say that passes through $0$ of two agents are {\em equivalent} if they 
occur during a period in which all agents are in phases $x: 3m/4<_m x<_m m/4$.
This notion defines a relation which is reflexive and symmetric.
For $m$ big enough by Lemma \ref{PC:move} with high probability
this relation is transitive and any two agents have equivalent passes through $0$.
Thus passes of agents through $0$ form {\em equivalence classes}.
This allows us to use argumentation similar to the one proposed in~\cite{AA+08},
however this time for the junta of leaders rather than for a single leader.

\begin{theorem}\label{PC:main}
Assume all agents start the phase clock protocol from the initial phase $0$ 
when at most $n^{1-\eps}$ leaders are already selected.
For any fixed $\eps,\eta,d > 0$, there exists a constant $m,$ 
s.t., the finite-state phase clock with parameter $m$ completes 
$n^\eta$ passes through~$0,$ s.t., the following conditions are satisfied
with high probability $1-n^{-\eta},$ for sufficiently large~$n$.
\begin{itemize}
\item

completes 
First $n^\eta$ passes through~$0$ of all agents form equivalence classes
in which each agent contributes exactly once and
the number of interactions between closest passes through $0$
in different equivalence classes is at least $d n \log n$.
\item
The number of interactions between two subsequent passes
through $0$ in any agent is $O(n \log n)$ whp.
\end{itemize}
\end{theorem}

\begin{proof}
By Lemma~\ref{PC:move} there exists $K,$ s.t., for $m=10K$ the thesis of this Lemma holds for the same $d>0$ as in the Theorem.
We consider ten subsets $A_0,A_1,A_2,\ldots,A_9$ of ${\mathbb Z}_{10K}$ defined as $A_i=\{iK,iK+1,\ldots,iK+_m K\}$.
By Lemma~\ref{PC:move} phases of all agents progress whp from $A_i$ to $A_{i+_{10} 1}$ (modulo 10)
in at least $d n\log n$ interactions whp. 
This implies that agents' passes through $0$ form equivalence classes whp 
and the number of interactions between closest passes through 0
in different equivalence classes is at least $d n \log n$ whp.
Since one way epidemic operates in $O(n\log n)$ interactions whp
each agent increments its phase in $O(n\log n)$ interactions.
Thus the number of interactions between two subsequent passes
through $0$ in any agent is $O(n \log n)$ whp.
\end{proof}




In conclusion, we formulate two useful facts related to phase clocks. 
Fact~\ref{PC:drop} states that if some leaders become followers during 
the phase clock protocol, then the phase clock can only slow down, 
but the upper bound on the number of interactions remains $O(n\log n)$.
Fact~\ref{PC:fault} states that any unsuccessful interactions can only 
slow down the phase clock.

\begin{fact}\label{PC:drop}
The reduction of the number of leaders during execution of the phase clock protocol 
can only slow down phase progression of agents on their clocks.
And if at least one agent remains as leader the number of interactions between two 
subsequent passes through $0$ in any agent is still $O(n\log n)$ whp.
\end{fact}

\begin{fact}\label{PC:fault}
If some interactions of the phase clock are faulty, i.e.,  they do not result 
in progression, then the phases of all agents do not become larger comparing 
to the protocol without faults.
\end{fact}

\section{Forming a junta}\label{FJ:sect}
In this section we describe $\formingjunta$ protocol.
The purpose of this protocol is to rapidly
elect from $n$ agents a junta of $O((n\log n)^{1/2})$ leaders 
assuming each agent utilises $O(\log\log n)$ states. 
This junta of leaders will be used to support phase clocks and 
eventual election of a unique leader.

The states of agents are represented as pairs $(l,a)$ where $a\in\{0,1\}$.
The value $l$ is a non-negative integer which we refer to as {\em level}.
During execution of the protocol agents with $a=0$ do not update their states.
However, any agent $v$ with value $a=1$ increments its level $l$ by $1$ or 
changes its value $a$ to $0$ during all interactions $v$ participates in.
The protocol stabilises when all agents conclude with $a=0$.
The transition function is defined, s.t., on the conclusion of this protocol there are
$O((n\log n)^{1/2})$ agents equipped with the highest computed value $l$ whp.
These agents form the desired {\em junta of leaders}.

All agents start in the same state $(l,a)=(0,1)$.
As agents in states $(l,0)$ do not get updated, we only need to specify how 
to update agents in states $(l,1)$ during pairwise interactions.
The transition function at level $l=0$ differs from levels $l>0$.
%
When an agent in state $(0,1)$ interacts with any agent in state $(0,1)$,
the final state of the initiator is $(1,1)$ and $(0,0)$ of the responder, i.e.,
\[(0,1),(0,1)\rightarrow(0,0),(1,1).
\]
When an agent $v$ in state $(0,1)$ interacts with any agent in state $(l,a),$ 
for levels $l>0,$ or with an agent in state $(0,0)$, 
the resulting state of $v$ is $(0,0)$.
If for any $l>0$ an agent $v$ in state $(l,1),$ participates in an interaction,
its state changes only if $v$ acts as the responder.
If the initiator is in state $(l',a)$ such that $l\leq l'$, 
the responder's state becomes $(l+1,1)$.
If the initiator is in state $(l',a)$ such that $l>l'$, 
the responder's state becomes $(l,0)$.

Let $B_l$ be the number of agents which reach level $l$ during execution of \formingjunta.
The value of $B_l$ depends on the execution thread of the protocol.
We first prove an upper bound on $B_1$.

\begin{lemma}\label{FJ:level0}
For $n$ large enough $1\leq B_1\leq n/2$.
\end{lemma}

\begin{proof}
During an interaction of two agents in states $(0,1)$ 
exactly half of the participating agents increase their level $l$ to $1$.
The remaining half ends up in state $(0,0)$ which becomes their final state.
During any other interaction in which an agent $v$ in state $(0,1)$ participates, 
$v$ changes its state to $(0,0)$. 
So at least half of the agents end up in state $(0,0)$.
Finally, since the first interaction of the protocol is between two agents in states $(0,1)$,
so at least one agent results in a state with $l>0$.
\end{proof}

Due to the reduction property of the protocol we have $B_1\geq B_2\geq B_3\geq B_4\geq\ldots$,
and in turn there exists the last $L$ for which value $B_L>0$.
We prove that $L=O(\log\log n)$ and in turn $B_L=O(\sqrt{n\log n}).$
We obtain this result by limiting values of $B_l,$ for all $l>1$.

\begin{lemma}\label{FJ:reductB}
Assume $n^{-1/3}\leq A<1$ and $B_l\leq A\cdot n$, 
then $B_{l+1}\leq \frac{11}{10} A^2\cdot n$ whp $1-e^{-n^{1/3}/300}$.
\end{lemma}

\begin{proof}
An agent $v$ contributing to value $B_l$ results in state $(l,1)$ as soon as it 
gets to level $l$ during the relevant interaction $t$.
Consider the first interaction $\iota$ succeeding $t$ in which $v$ acts as the responder.
During this interaction the initiator is on level $l'\geq l$ with probability $p_{\iota}\leq B_l/n\leq A.$
Thus $v$ moves to level $l+1$ with probability at most $A$ as
otherwise the responder would end up in state $(l,0)$ and would not contribute to $B_{l+1}$. 
Consider now the sequence of all $B_l$ interactions $\iota$, in which agents in state $(l,1)$ act as responders.
We can attribute to this sequence a binary $0$-$1$ sequence $\sigma$ of length $B_l$, s.t.,
if during interaction $\iota$ an agent ends up in state $(l+1,1)$ the respective entry in $\sigma$ becomes 1.
Otherwise, this entry becomes 1 with probability $(1-A)/(1-p_{\iota})$ and 0 with probability $A/(1-p_{\iota})$.
Thus the probability of each entry being $1$ is independently equal to $A$ and
the number of $1$s in $\sigma$ is at least $B_{l+1}$. The expected number of these $1$s is $A\cdot B_l\leq A^2 n$.
By the Chernoff bound $B_{l+1}> \frac{11}{10}A^2 \cdot n$ with probability at most $e^{-A^2 n/300}<e^{-n^{1/3}/300}$.
\end{proof}

\begin{lemma}\label{FJ:zeroB}
If $B_l\leq n^{1/3}$ we get $B_{l+1}>0$ with probability at most $n^{-1/3}$.
\end{lemma}

\begin{proof}
If $B_l\leq n^{1/3},$ the probability for any agent on level $l$ to
get to level $l$ is at most $n^{-2/3}$. Thus by the Union bound the probability
of some agent getting to level $l$ is at most $n^{-1/3}$.
\end{proof}

\begin{lemma}\label{FJ:lastsmall}
There exists a constant $c>0,$ for which if $B_l\geq c\sqrt{n\log n},$
the probability of $B_{l+1}=0$ is negligible.
\end{lemma}

\begin{proof}
Consider a group of $c\sqrt{n\log n}/2$ agents which move to level $l$ after
this level is already reached by $c\sqrt{n\log n}/2$ other agents.
Any agent in this group moves to level $l+1$ with probability at least $c\sqrt{\log n/4n}.$
Since all these agents advance to level $l+1$ independently,
the probability that $B_{l+1}=0$ is at most
\[ \left(1-c\sqrt{\log n/4n}\right)^{c\sqrt{n\log n}/2}<
      e^{-c^2\log n/4}<n^{-c^2/4}.
\]
This last value is smaller than $n^{-\eta},$ for $c$ large enough. 
\end{proof}

\begin{theorem}\label{FJ:main}
In protocol $\formingjunta$ the largest level $L$ for which $B_L>0$ 
satisfies $L=\log\log n+c$ for some constant $c$ and $B_L=O(\sqrt{n\log n})$ whp.
\end{theorem}

\begin{proof}
By Lemma \ref{FJ:level0} we have $B_1\leq n/2$.
By Lemma \ref{FJ:reductB} we conclude $B_2\leq \frac{11}{10}\cdot\frac{n}{4}$ whp $1-e^{-n^{1/3}/300}$.
Furthermore, $B_3\leq (\frac{11}{10})^3\cdot \frac{n}{2^4}$ whp $1-2e^{-n^{1/3}/300}$.
And in general $B_l\leq (\frac{11}{10})^{2^l-1}\cdot {n/2^{2^l}}$ whp $1-le^{-n^{1/3}/300}$.
Thus for $L'=\log\log n+1$ we get $B_{L'}\leq n^{1/3}$ whp.
Further on by
Lemma \ref{FJ:zeroB} the value $B_{L''},$ where $L''=L'+c$, equals $0,$ 
for some constant $c,$ whp, and we have $L<L''$.
By Lemma \ref{FJ:lastsmall} on the last level $L$ for which $B_L>0$
we have $B_L=\O(\sqrt{n\log n})$ whp.
Thus both conditions hold whp.
\end{proof}

The last lemma bounds from above the running time of protocol $\formingjunta$.

\begin{lemma}\label{FJ:time}
The protocol $\formingjunta$ stabilises in $O(n\log n)$ iteractions whp.
\end{lemma}

\begin{proof}
Recall from Lemma \ref{FJ:level0} that $B_1\leq n/2$ 
and the number of agents with the final state $(0,0)$ is at least $n/2$.
Each agent in this group ends up in this state during its first interaction.
Since every agent interacts at least once during the first $O(n\log n)$ 
interactions of the protocol whp,
all agents ending up in state $(0,0)$ do so during this time whp.
One can show that an agent does not experience an interaction 
during the first $cn\ln n$ interactions, for a constant $c,$ with probability 
\[\left(1-\frac{2}{n}\right)^{cn\ln n}\leq n^{-2c}. 
\]
Thus there exists a positive constant $c$ for which after $cn\ln n$ 
interactions each agent experiences its first interaction whp $1-n^{-\eta-1}$.
Any agent that interacts as the responder with an agent in state $(0,0)$
sets its value $a$ to $0$ which concludes the transition process. 
And after at least $n/2$ agents are in state $(0,0)$, the probability that the current interaction 
is one of such interactions w.r.t. a particular responder is at least $\frac{1}{4n}$.
Thus the probability that a given agent does not have $a=0$ after $c'n\ln n$ iterations is
\[\left(1-\frac{1}{4n}\right)^{c'n\ln n}\leq n^{-c'/4}, 
\]
and for the constant $c'$ big enough $n^{-c'/4}<n^{-\eta-1}$.
Thus the number of interactions needed to obtain $a=0$ in all agents is $O(n\log n)$ whp.
\end{proof}

Finally we prove a corollary stating that ``spoiling'' (for the definition check below) protocol $\formingjunta$
does not affect validity of statements of Theorem \ref{FJ:main} and Lemma~\ref{FJ:time}.
Using the notion of a spoiled protocol instead of the flawless one is needed to bound
the total number of states in the leader election protocol to $O(\log\log n)$.
Let {\em spoiled} $\formingjunta$ protocol be any protocol obtained by 
changing spontaneously some 
states  from $(l,a)$ to $(0,0)$, where $l$ is not the highest level reached so far in the population.
We denote the total numbers of agents that reach level $l$ in this spoiled protocol by $B^*_l,$
and the highest level for which $B^*_l>0$ by $L^*$.
Observe that in the spoiled protocol all agents at level $L^*$ never go through state $(0,0)$.

\begin{corollary}\label{FJ:final}
Level $L^*$ satisfies the condition $L^*=O(\log\log n)$ and $B^*_{L^*}=O(\sqrt{n\log n})$ whp.
Also spoiled $\formingjunta$ protocol stabilises after $O(n\log n)$ interactions whp.
\end{corollary}

\begin{proof}
The numbers of agents $B^*_l$ reaching level $l$ in the spoiled protocol are not 
larger respectively than numbers $B_l$ from the flawless protocol, thus $L^*=O(\log\log n)$.
Also Lemma \ref{FJ:lastsmall} still bounds from above $B^*_{L^*}$ by $O(\sqrt{n\log n})$ whp.
Thus the running time of the spoiled protocol is not larger than the flawless one.
\end{proof}

\section{Leader election}\label{LE:sect}
\newcommand{\draw}{\mbox{\tt draw}}
\newcommand{\spread}{\mbox{\tt spread}}

In this section we describe how to combine the protocols described in the two previous 
sections to obtain a new fast population protocol for leader election.
This new leader election protocol operates in (parallel) time $\Theta(\log^2 n)$ on 
populations with agents equipped with $\Theta(\log\log n)$ states.

The new leader election protocol assumes that at the beginning 
there is a non-empty subset (possibly the whole population) of agents 
which are candidates for leaders, and this subset is gradually reduced to a singleton.
The protocol consists of $\Theta(\log n)$ repetitions of the external loop, each formed 
of $\Theta(n\log n)$ interactions controlled by the ordinary mode of the phase clock.
This emulates a leader election scheme starting from a set of leader candidates in which
during each repetition every candidate picks independently at random 
either bit $0$ or $1$ by tossing a fair coin.
In real terms, the coin tossing process relies on the initiator versus responder 
selection performed by the random scheduler.
The candidates which pick $1$ broadcast message "$1$" to all other agents.
And when a candidate with chosen $0$ receives message "$1$" 
it stops being a candidate for the leader.

\begin{theorem}\label{LE:scheme}
The scheme proposed above selects a unique leader 
during $\Theta(\log n)$ repetitions whp.
\end{theorem}

\begin{proof}
If the number of candidates is at least 2, the probability that in the relevant 
repetition of the external loop at least half of the candidates draw $0$ is at least $1/2$.
Consider a series of $c\log n$ consecutive repetitions and form a 
binary $0$-$1$ sequence $\sigma$ of length $c\log n,$ 
in which the entries correspond to these repetitions.
If prior to a repetition only one candidate remains,  the entry in $\sigma$ 
is chosen uniformly at random by a single coin~toss.
If there are more candidates drawing $1$s than $0$s, 
then the relevant entry becomes~$1$.
If there is more than one candidate and at least half of them draw $0$, en extra random
selection is triggered, s.t., the probability of choosing $0$ is exactly $1/2$.
Note, that if the sequence has at least $\log n$ $1$s, then exactly one leader remains.
By the Chernoff bound the probability, that $\sigma$ contains less than $\log n$ $1$s 
is smaller than $e^{-(1-1/c)^2c\log n/2},$ and in turn smaller 
than $n^{-\eta},$ for a constant $c$ large enough.
\end{proof}

The main problem with utilisation of the protocol described above is the need of 
implementing a counter of $\Theta(\log n)$ repetitions utilising a very small memory.
We also need to implement multi-broadcast of $1$s which 
requires $\Theta(n\log n)$ interactions whp.
The multi-broadcast can be implemented via one-way epidemic 
described in Section~\ref{PR:sect}.
The two processes can be controlled by the phase clock run in both the external and 
the ordinary modes respectively, using a constant number of states.
This is conditioned by forming a junta of at most $n^{1-\eps}$ leaders.
In Section~\ref{FJ:sect} we described the relevant $\formingjunta$ protocol
which reduces the number of leaders to $O(\sqrt{n\log n})$ and which
utilises $\Theta(\log\log n)$ states at each agent.
Our leader election protocol starts with a single execution of 
protocol $\formingjunta$ which is followed by the leader reduction
mechanism allowing to reduce the original junta team to a single leader.

All agents enter the leader election protocol in the same state.
The current state of an agent is represented by a vector $(l,a,b,x,y,\mathbf{z})$ where
all entries, with exception of $l,$ have constant size descriptions.
A non-negative integer $l$ refers to the number of levels bounded by $O(\log\log n)$.
Other positions contain small integer constants 
$a\in\{0,1\}$, $b\in\{\lead,\foll\}$, which refer to the leadership status, and
$x,y\in{\mathbb Z}_m$ are utilised for the phase clock's ordinary and external modes respectively.
The remaining state overheads imposed by our protocol are encoded in $\mathbf{z}=(z_0,z_1,z_2)$ 
which is limited to a constant number of values used to steer the protocol of leader elimination.
Here $z_0\in\{\draw,\spread\}$ means participation in either drawing 0,1 by leaders or spreading value 1 if was drawn.
Since $l$ assumes $O(\log\log n)$ values and all other variables can have only a constant number of values,
the total number of states in the protocol is $O(\log\log n)$.
This number of states can be more precisely estimated since $l$ can be upperbounded by $\log\log n+c$,
where $c$ is a constant depending on $\eta$.
Indeed a more careful estimation based on Theorem \ref{FJ:main} gives the number of levels $\log\log n+1$
needed to reduce the number of agents to $n^{1/3}$ whp.
And complete elimination of agents would require extra constant $c=\Omega(\eta)$ levels whp $1-n^{-\eta}$.
Taking into account the number of possible values of variables $a,b,x,y,\mathbf{z}$ we get the total
number of states $48 m^2(\log\log n+c)$ for $m,c$ depending on $\eta$.


\begin{algorithm}[ht!]
	\caption{{\tt Leader\_election}($A$:$(l,a,b,x,y,\mathbf{z}))$}
	\label{alg:unknown}
	\begin{algorithmic}[1]
    \State {\bf execute} \formingjunta\ \Comment{concludes with $a=0$}
    \State $(x,y,z_0,z_1,z_2)\gets (0,0,\draw,\emptyset,0)$ \Comment{set entries}
 \Loop\ through all interactions\Comment{main loop}
    \State $\{A$ meets an agent in state $(l',0,b',x',y',\mathbf{z'})\}$ 
    \State{\bf if} $(l<l')$ {\bf then} 
    \State\ \ \ {$(l,b,x,y,z_0,z_1,z_2)\gets (l',\foll,0,0,\draw,\emptyset,0)$} 
    \State{\bf if} ($A$ is responder) {\bf then} 
    \State\ \ \ perform operations of phase clocks
    \State{\bf if} (phase $x$ just passed through 0) {\bf then}
    \State\ \ \ {\bf if} ($z_0=\draw$) {\bf then} 
    \State \ \ \ \ \ \ $(z_0,z_2)\gets(\spread,0)$
    \State \ \ \ {\bf else} $(z_0,z_1)\gets(\draw,\emptyset)$
    \State{\bf if} ($z_0=\draw\ \&\ z_1=\emptyset\ \&\  $
    \State $ l=\lead\ \&\ l'=\foll$) {\bf then}
    \State\ \ \ {\bf if} ($A$ is responder) {\bf then}
    \State\ \ \ \ \ \ $z_1\gets 0$
    \State\ \ \ {\bf else} $z_1\gets 1$
    \State{\bf if} ($z_0=\spread\ \&$ $A$ is responder) {\bf then}
    \State\ \ \ $z_2\gets\max\{z_2,z_1',z_2'\}$       
    \State{\bf if} ($z_0=\spread$ \& $l=\lead$ \& 
    $z_1=0$ \& $z_2=1$) {\bf then}
    \State\ \ \ $l\gets\foll$ 
  \EndLoop
    \end{algorithmic}
\end{algorithm}


\paragraph{Spoiled \formingjunta\ protocol.} All agents start the leader election protocol with $(l,a,b,x,y)=(0,1,\lead,0,0),$ and they run \formingjunta\ protocol in state 
$(l,a),$ for as long as $b=\lead$.
As soon as $b$ gets value $\foll,$ which is irreversible, 
according to \formingjunta\ protocol the state of the relevant agent becomes $(0,0)$.
This happens only when $l$ is not at the highest level in the population, 
thus the protocol \formingjunta\ gets spoiled this way only occasionally.
The relevant detail will be described in the next paragraph.
According to Corollary~\ref{FJ:final} each agent should conclude spoiled \formingjunta\ protocol 
in the first $\Theta(n\log n)$ interactions whp.

\paragraph{Phase clocks on different levels.}
Once value $a$ becomes $0$, the agent starts its phase clocks on level $l$ as 
the leader with parameters $x=y=0$.
When an agent with phase clock on level $l$ interacts with an agent with the phase clocks on a higher level $l'>l$,
its state is rewritten $(l,b,x,y)\leftarrow(l',\foll,0,0)$.
This way the agent aligns its phase clocks in phase $0$ on level $l'$ and 
ends up in state $(0,0)$ in the spoiled variant of \formingjunta\ protocol.
The level of the phase clock can be incremented this way many times 
until it attains the maximum level $L^*$ ever reached by the population.
Thus eventually all agents run together the phase clock on level $L^*$.
All agents which advance to level $L^*$ in spoiled \formingjunta\  protocol 
are the leaders of the phase clocks and others act as followers.
Note that agents that end \formingjunta\  protocol on level $L^*$ are
never notified that they computed the highest level.
Only agents that end \formingjunta\ on a lower level are forced in some
moment to join the protocol with level $L^*$ as followers.
We run the phase clock in the ordinary mode and in the external mode simultaneously to implement the two loops
described in the beginning of Section~\ref{PC:sect}.
The phase clock in the ordinary mode is driven by all interactions in which the responder has value $a=0$.
If the responder interacts with an initiator on a higher level it advances its clock level as described above.
If the responder has the same clock level as the initiator, they both perform one interaction in the ordinary mode.
If the responder interacts with the initiator on a lower level or having $a=1$, then this interaction is void in the ordinary mode.
The phase clock operates in the ordinary mode until it passes through $0$ for the first time.
And it counts for each agent the first $\Theta(n\log n)$ interactions by Fact~\ref{PC:fault}. 

\paragraph{Random coin tosses.} Each remaining leader $v$ picks randomly $0$ or $1$ during 
the first interaction with a non-leader after the phase of $v$ (in the ordinary mode of the clock) passes through $0$.
If the non-leader is the initiator $v$ chooses $1,$ otherwise $v$ picks $0$.
This gives a truly random value to each leader, and since there are $O(\sqrt{n\log n})$ leaders,
this process is completed whp during a constant number of interactions.

\paragraph{Leader candidate elimination on the highest level.}
After choosing a value $0$ or $1$ at random, the leaders multi-broadcast $1$s to the whole population
via one-way epidemic. The required $\Theta(n\log n)$ interactions 
are counted with the help of the phase clock in the ordinary mode.
In order to obtain a unique leader whp, this process is iterated $\Theta(\log n)$ times by the external loop
and controlled by the phase clock in the external mode.
The protocol concludes at each agent, when its external clock attains phase $m-1$.
The following theorem holds.

\begin{theorem}\label{LE:whp}
The protocol described above finds a unique leader in $O(n\log^2 n)$ interactions whp.
\end{theorem}

Now we formulate a Las Vegas variant of our algorithm to more accurately 
match the existing lower bound $\Omega(\log\log n)$ on the number of states in fast leader election~\cite{AA+17}.

\begin{theorem}
For agents equipped with $O(\log\log n)$ states, there exists a leader election protocol 
which always gives the correct answer and works in parallel time $O(\log^2 n)$ whp.
\end{theorem}

\begin{proof}
In the Las Vegas variant of our protocol 
the external clock utilises the set of transitions defined as before.
However, we replace ${\max}_m$ by the standard maximum as we 
assume that the external clock stops after reaching phase $m-1.$
We also need to impose here the level limit $L'=\Theta(\log\log n)$ in \formingjunta\ protocol.
If this level $L'$ is achieved, which occurs with a negligible probability,
the agent's level is no longer incremented as it plays the role of $L^*$.
In a very unlikely event an agent $v$ may interact with any other agent with {\em a distant} ordinary phase clock value,
i.e., the relevant phase clock values $x$ and $x+_m a$ satisfy $m/5<a<4m/5$.
In such case we make agent $v$ utilise all subsequent interactions in its external clock as meaningful. 
In addition, 
it also infects with this setting all other agents via one-way epidemic. 
By Theorem~\ref{LE:whp} we can construct a fast leader election protocol with the clock 
phases drawn from ${\mathbb Z}_m$, s.t., a single leader is elected and the external phase clocks 
in all agents conclude in phase $m-1$ during the first $O(n\log^2 n)$ interactions whp $1-n^{-10}.$ 
Thanks to Lemma~\ref{PC:move} used in the proof of correctness of the relevant clock construction
we can derive an extra observation that no two agents can have distant ordinary phase clock values 
during execution of the protocol whp $1-n^{-10}$.

If a leader enters external phase $m-1$ in the fast protocol we have just described, it can no longer be 
eliminated by this protocol. Independently, all agents run from the beginning a slow two-state based leader 
election protocol with the expected number of interactions $O(n^2\log n)$~\cite{DS+15}.
In this slow protocol, whenever two leader candidates interact directly the initiator eliminates the responder.
If a leader candidate of this slow protocol reaches phase $m-1$ in the external mode,
it stops being a candidate for the leader, unless it is still a leader in the fast protocol.
The leaders reaching external phase $m-1$ in the external clock eliminate other leaders in the fast protocol in 
direct pairwise interactions according to the slow protocol principle.

Note that all agents complete \formingjunta\ protocol with expected $O(n\log n)$ interactions.
Assume this part of leader election is already completed.
Let $E$ be the expected number of interactions in the leader election algorithm. 
We have
\[E\leq (1-n^{-10})\cdot cn\log^2 n+n^{-10}\max\{E',E''\}.
\]
In this formula $E'$ and $E''$ are the expected numbers of interactions if we start from the worst case 
configurations respectively not containing ($E'$) and containing ($E''$) distant ordinary clock phases.
If~we start from the configuration not containing distant ordinary clock phases,
the external phase clock reaches phase $m-1$ in all agents or all leaders disappear
during $O(n\log^2 n)$ interactions whp $1-n^{-10},$ unless an interaction between agents with 
distant ordinary clock phases occurs at some point.
This can be proved using Lemma~\ref{PC:move} and argument analogous to the proof of Theorem \ref{LE:whp}.
In the latter case the external clock reaches phase $m-1$ whp in $O(n\log n)$ interactions 
(after a distant interaction takes place) unless all leaders in the fast protocol disappear.
When the fast leader election protocol fails, i.e., 
it either produces multiple leaders or all candidates for leaders disappear,
the leader election process is completed during $O(n^2\log n)$ interactions 
of the slow leader election protocol. Thus
\[E'\leq (1-n^{-10})\cdot cn^2\log n+n^{-10}\max\{E',E''\}.
\]
If $E'\geq E''$ we get $E',E''=O(n^2\log n)$ from this inequality.
When we start in the worst case configuration in which there are two agents with distant ordinary 
phase clock values, they meet during the first interaction of the protocol with probability at least $n^{-2}$. 
And when this happens, the external clock reaches phase $m-1$ in $O(n\log n)$ interactions whp
and also in this case the unique leader is selected whp during $O(n^2\log n)$ interactions of the slow protocol. 
Thus
\[E''\leq n^{-2}\cdot cn^2\log n+(1-n^{-2})(\max\{E',E''\}+1).
\]
If $E''\geq E'$, we get $E',E''=O(n^2\log n)$ from this inequality.
And since $E',E''=O(n^2\log n)$ we conclude $E=O(n\log^2 n)$.
\end{proof}

\section{Conclusion}\label{FD:sect}

We studied in this paper fast and space efficient leader election in population protocols. 
Our new protocol stabilises in (parallel) time $O(\log^2 n)$ when each agent is equipped with $O(\log\log n)$ states. 
This double logarithmic space utilisation matches asymptotically the lower bound $\frac{1}{2}\log\log n$ on 
the minimal number of states required by agents in any leader election algorithm with the running 
time $o(\frac{n}{{\rm polylog}\ n})$, see~\cite{AA+17}.
For the convenience of the reader we provide the logical structure of the full argument 
in the form of a diagram, see Figure~\ref{f:diagram}.
\begin{figure}[thb]
\includegraphics[scale=0.2150]{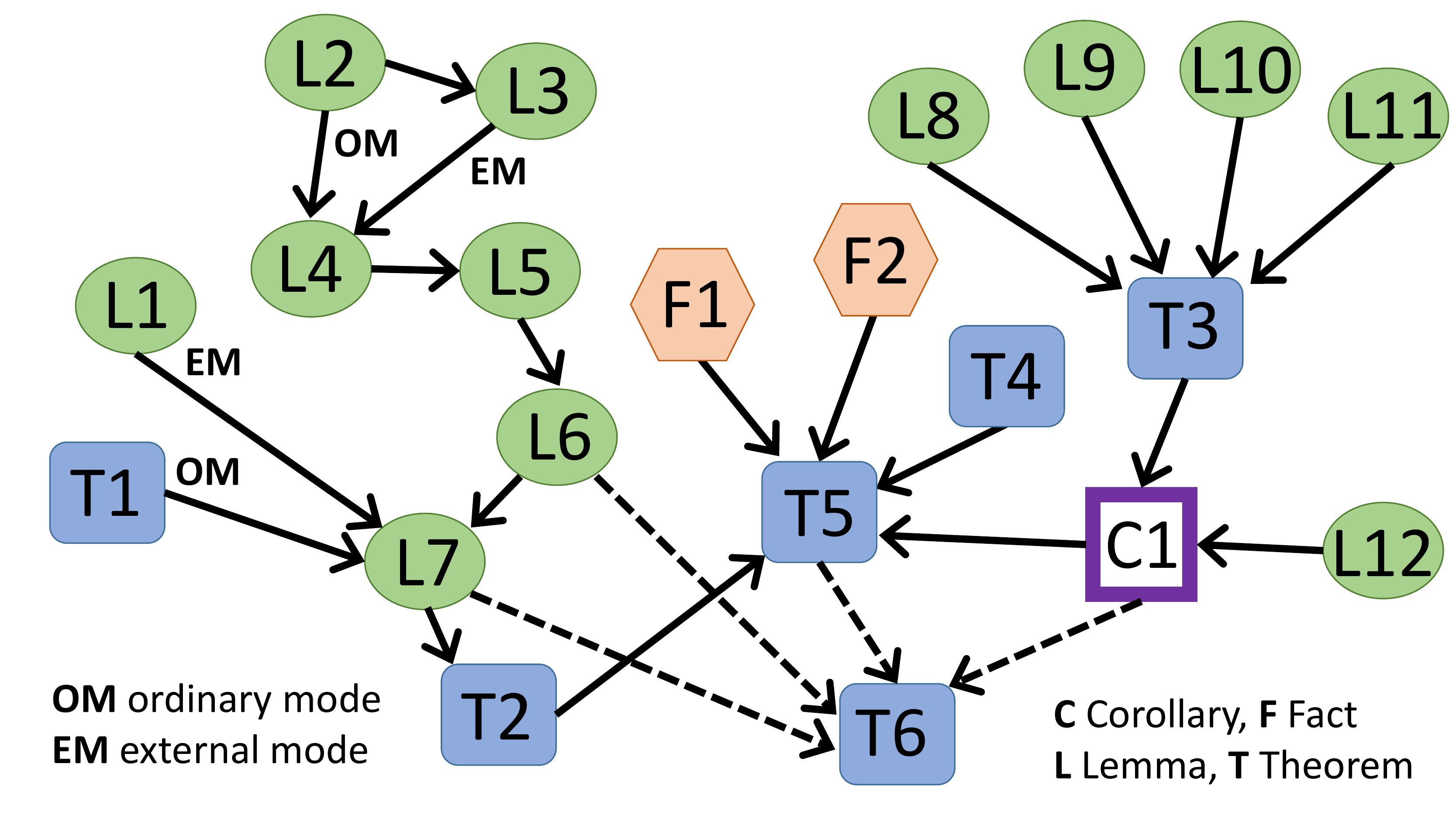}
\caption{The structure of the argument\label{f:diagram}}
\end{figure}
We also share with the reader a diagram illustrating transitions 
between states during leader election protocol, see Figure~\ref{f:diag2}.

\begin{figure}[bth]
\includegraphics[scale=0.240]{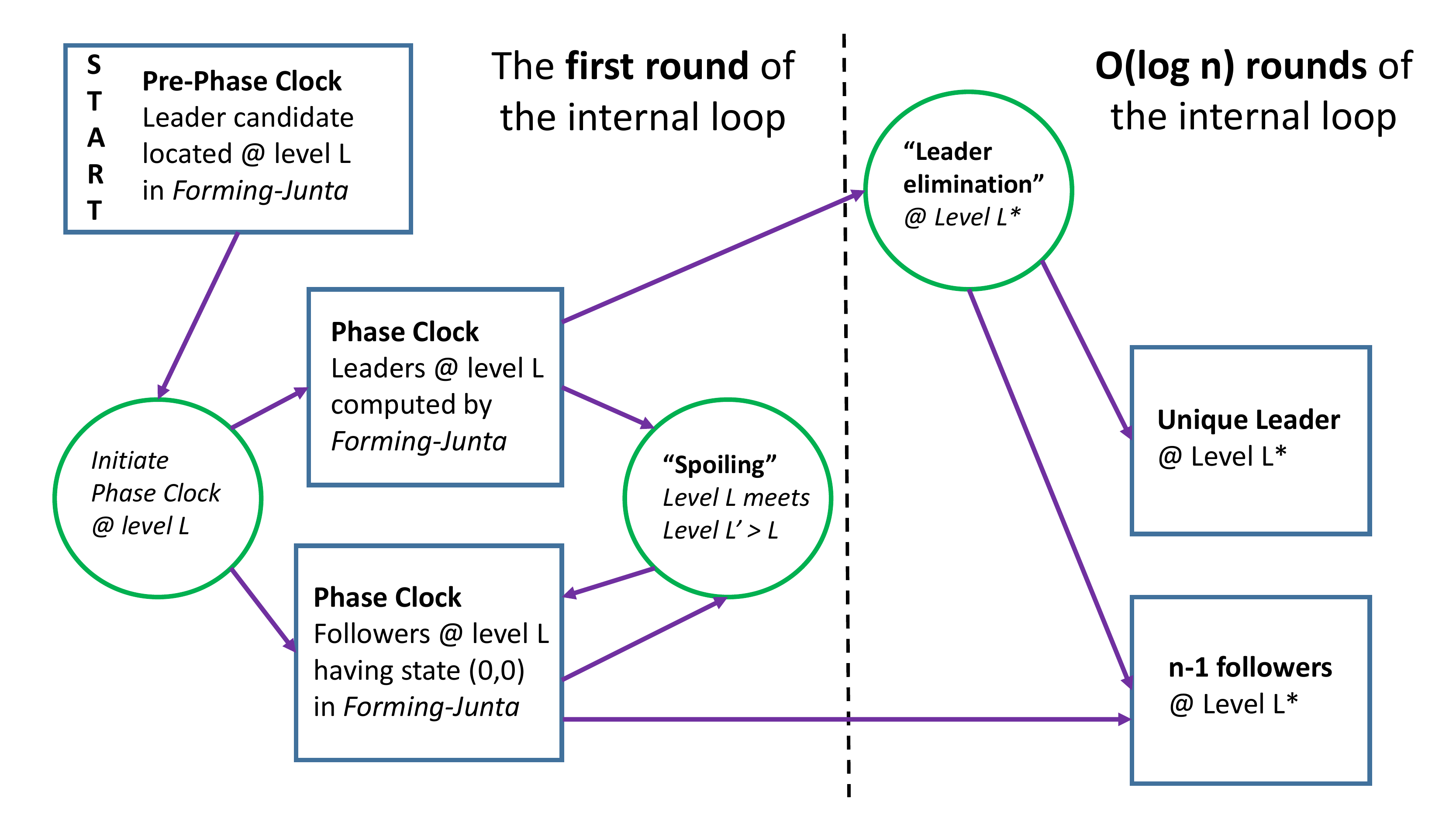}
\caption{Transitions between states\label{f:diag2}}
\end{figure}

\paragraph{Further extensions.} In this paper we propose a Las Vegas type algorithm which achieves 
stabilisation in the sense, that eventually a single (unique) agent arrives in a leader state 
and all other agents arrive in follower states. 
More precisely, in some (unlikely) scenarios the chosen leader can hover between different leader states 
and similarly all other agents can switch between different follower states indefinitely.
The stabilisation process can be also considered in a stronger sense where the states of all agents finally freeze as pointed out to us by Dominik Kaaser. 
In what follows we explain briefly how our algorithm can be modified to meet this stronger requirement.  
%

In the enhanced variant of our algorithm on the conclusion of spoiled \formingjunta\ protocol 
we run three different phase clocks including:
(1) the ordinary clock run by all leaders computed by \formingjunta\ protocol;
(2) an external clock run by all leaders computed by \formingjunta\ protocol; and
(3) an external clock run by all leaders that remain {\em active}.
%

As in the Las Vegas protocol $\max_m$ is replaced by $\max$ in external clock protocols.
Similarly to the protocols proposed in this paper the leaders are eliminated using coin tossing.
The outcome of coin tosses is communicated via one-way epidemic protocols controlled 
by the ordinary clock (1), and the number of the relevant loop repetitions is controlled by clock (3).
Once a leader makes a losing coin toss (picks value 0), it is not instantly eliminated.
Instead, the leader becomes {\em inactive} and awaits any motion of clock (3).
When clock (3) eventually moves the leader gets eliminated.
When clock (2) reaches the last phase, all agents reach the final state associated with the clock level.
This process guarantees that at least one leader survives.
Finally if there is more than one leader in the final state,
the remaining leaders elect a unique leader 
using the slow leader election protocol. 

This process finishes in expected time $\Theta(\log^2 n)$ assuming flawless performance of clock (1).
A problem may occur in an unlikely event when there are agents for which phases on clocks of type (1) are distant.
And indeed if during an interaction of two agents their respective phase clocks of type (1) are in distant phases,
these two agents inform all other agents (including leaders) about broken (out of phase) clocks 
via one-way epidemic. Thus all agents get to the final state associated with their clock level.
As in the previous protocol if an agent on a given level (even in the final state) interacts with
an agent on a higher level, it switches to this level resetting its phase clocks.

\paragraph{Open problems.}There are some interesting unanswered questions left for further consideration.
For example, whether one can select whp a unique leader in time $o(\log^2 n)$ with $O(\log\log n)$ states available at each agent.
Another question refers to the exact space complexity of majority as well as plurality consensus in deterministic population protocols
considered recently, see, e.g.,~\cite{GH+16}.

\paragraph{Acknowledgements.} We would like to thank our research collaborators
Tomasz Jurdzi\'nski, Aris Pagourtzis, Tomasz Radzik, Micha\l\ R\'o\.za\'nski, and Paul Spirakis for their suggestions in early stages 
of this work, as well as Dave Doty, Rati Gelashvili and Dominik Kaaser for their helpful comments on earlier versions of this paper.
Special thanks go to all anonymous reviewers whose exhaustive comments helped us to improve the quality of the final version.

%


\bibliographystyle{plain}

\begin{thebibliography}{10}

\bibitem{AR16}
A.~Andoni and I.P.~Razenshteyn,
Tight Lower Bounds for Data-Dependent Locality-Sensitive Hashing, 
Proc. {\em Symposium on Computational Geometry} 2016, paper 9:1-9:11.

\bibitem{A80}
D.~Angluin, 
Local and global properties in networks of processors,
Proc. {\em 12th Annual ACM Symposium on Theory of Computing}, STOC 1980, 82--93.

\bibitem{AA+04}
D.~Angluin, J.~Aspnes, Z.~Diamadi, M.J.~Fischer, and R.~Peralta, 
Computation in networks of passively mobile finite-state sensors. 
Proc. {\em 23rd Annual ACM Symposium on Principles of Distributed Computing},
PODC 2004,~290--299.

\bibitem{AA+06}
D.~Angluin, J.~Aspnes, Z.~Diamadi, M.J.~Fischer, and R.~Peralta, 
Computation in networks of passively mobile finite-state sensors. Distributed Computing, 
18(4), 2006, 235--253.

\bibitem{AA+08}
D.~Angluin, J.~Aspnes, D.~Eisenstat.
Fast computation by population protocols with a leader, 
Distributed Computing 21(3), 2008, 183--199.

\bibitem{AA+08b}
D.~Angluin, J.~Aspnes, and D.~Eisenstat. 
A simple population protocol for fast robust approximate majority, 
Distributed Computing, 21(2), 2008, 87--102.

\bibitem{AA+17} 
D.~Alistarh, J.~Aspnes, D.~Eisenstat, R.~Gelashvili and R.L.~Rivest,
Time-Space Trade-offs in Population Protocols,
Proc. {\em 28th Annual ACM-SIAM Symposium on Discrete Algorithms}, SODA 2017, 2560--2579.

\bibitem{AAG17}
D.~Alistarh, J.~Aspnes, and R.~Gelashvili,
Space-Optimal Majority in Population Protocols,
Proc. {\em 29th Annual ACM-SIAM Symposium on Discrete Algorithms}, SODA 2018, 
also {\tt arXiv:1704.04947 [cs.DC]}. 

\bibitem{AG+15}
D.~Alistarh and R.~Gelashvili,
Polylogarithmic-time leader election in population protocols,
Proc. {\em 42nd International Colloquium on Automata, Languages, and Programming}, ICALP 2015, 479--491.

\bibitem{AD+91}
A.~Arora, S.~Dolev, and M.G.~Gouda, 
Maintaining digital clocks in step, 
Proc. {\em 5th International Workshop on Distributed Algorithms}, (LNCS 579) 1991,~71--79.

\bibitem{AS91}
H.~Attiya and M.~Snir, 
Better Computing on the Anonymous Ring, 
Journal of Algorithms 12, 1991, 204--238.

\bibitem{AS+88}
H. Attiya, M. Snir, and M. Warmuth, 
Computing on an Anonymous Ring, 
Journal of the ACM 35, 1988, 845--875.

\bibitem{AW04}
H. Attiya and J.~Welch, 
Distributed Computing: Fundamentals, Simulations, and Advanced Topics, 
2nd Edition, Willey, April 2004.

\bibitem{BC+15}
L.~Becchetti, A.E.F.~Clementi, E.~Natale, F.~Pasquale, and R.~Silvestri,
Plurality consensus in the gossip model, 
Proc. {\em 26th Annual ACM-SIAM Symposium on Discrete Algorithms}, SODA 2015,371--390.


\bibitem{BF+16}
P.~Berenbrink, T.~Friedetzky, G.~Giakkoupis, and P.~Kling. 
Efficient plurality consensus, or: The benefits of cleaning up from time to time. 
Proc. {\em 43rd International Colloquium on Automata, Languages, and Programming}, ICALP 2016, 1--14.

\bibitem{BS+96}
P.~Boldi, S.~Shammah, S.~Vigna, B.~Codenotti, P.~Gemmell, and J.~Simon, 
Symmetry Breaking in Anonymous Networks: Characterizations, 
Proc. {\em 4th Israel Symposium on Theory of Computing and Systems}, ISTCS 1996, 16--26.

\bibitem{BV99}
P.~Boldi and S.~Vigna, 
Computing Anonymously with Arbitrary Knowledge, 
Proc. {\em 18th ACM Symp. on Principles of Distributed Computing}, PODC 1999, 181--188.

\bibitem{B80}
J.E. Burns, 
A Formal Model for Message Passing Systems, 
Tech. Report TR-91, Computer Science Department, Indiana University, Bloomington, September 1980.

\bibitem{CKL16}
L.~Cardelli, M.~Kwiatkowska, and L.~Laurenti,
Programming Discrete Distributions with Chemical Reaction Networks, 
Proc. {\sl 22nd International Conference DNA Computing and Molecular Programming},
DNA 2016,~35--51.

\bibitem{CM+11}
I.~Chatzigiannakis, O.~Michail, S.~Nikolaou, A.~Pavlogiannis, and P.G.~Spirakis,
Passively mobile communicating machines that use restricted space. 
Proc. {\em 7th ACM ACM SIGACT/SIGMOBILE International Workshop on Foundations of Mobile
Computing}, 2011, 6--15.

\bibitem{CC+14}
H.-L.~Chen, R.~Cummings, D.~Doty, and D.~Soloveichik, 
Speed faults in computation by chemical reaction networks,
Distributed Computing, Springer 2014, 16--30.

\bibitem{DP14}
D.~Dereniowski and A.~Pelc, 
Leader election for anonymous asynchronous agents in arbitrary networks, 
Distributed Computing 27, 2014, 21--38.

\bibitem{DP04}
S.~Dobrev and A.~Pelc, 
Leader Election in Rings with Nonunique Labels, 
Fundamenta Informaticae 59, 2004, 333--347.

\bibitem{DW04}
S.~Dolev and J.L.~Welch, 
Self-stabilizing clock synchronization in the presence of Byzantine faults, 
Journal of the ACM, 51(5), 2004, 780--799.

\bibitem{D14}
D.~Doty, 
Timing in chemical reaction networks. 
Proc. {\em 25th Annual ACM-SIAM Symposium on Discrete Algorithms}, SODA 2014, 772--784. SIAM.

\bibitem{DD+03}
A.~Daliot, D.~Dolev, and H.~Parnas,
Self-stabilizing pulse synchronization inspired by biological pacemaker networks,
Proc. {\em Self-Stabilizing Systems}, SSS 2003,~32--48.

\bibitem{DS+15}
D.~Doty and D.~Soloveichik,
Stable leader election in population protocols requires linear time, 
Proc. {\em 29th International Symposium on Distributed Computing}, DISC 2015, 602--616.

\bibitem{FK+04}
P.~Flocchini, E.~Kranakis, D.~Krizanc, F.L.~Luccio and N.~Santoro, 
Sorting and Election in Anonymous Asynchronous Rings, 
Journal of Parallel and Distributed Computing 64, 2004, 254--265.
responder
\bibitem{FL87}
G.N.~Fredrickson and N.A.~Lynch, 
Electing a Leader in a Synchronous Ring, 
Journal of the ACM 34, 1987, 98--115.

\bibitem{FP12}
E.~Fusco and A.~Pelc, 
Knowledge, level of symmetry, and time of leader election, 
Proc. {\em 20th Annual European Symposium on Algorithms}, ESA 2012, 479--490.

\bibitem{FP11} 
E. Fusco, A. Pelc, 
Trade-offs between the size of advice and broadcasting time in trees, 
Algorithmica 60, 2011, 719--734.


\bibitem{GH+15}
L.~G\k asieniec, D.D.~Hamilton, R.~Martin, and P.G.~Spirakis,
The match-maker: Constant space distributed majority via random walks, 
Proc. {\em 17th International Symposium on Stabilization, Safety, and Security of Distributed Systems}, SSS 2015,  67--80.

\bibitem{GH+16}
L.~G\k asieniec, D.~Hamilton, R.~Martin, P.G.~Spirakis and G.~Stachowiak,
Deterministic Population Protocols for Exact Majority and Plurality,
OPODIS'16 (post-conference proceedings, to appear).

\bibitem{GP16}
M.~Ghaffari and M.~Parter, 
A polylogarithmic gossip algorithm for plurality consensus.
Proc. {\em 34th Annual ACM SIGACT-SIGOPS Symposium on Principles of Distributed Computing}, PODC 2016, 117--126.

\bibitem{GM+16}
Ch.~Glacet, A.~Miller, and A.~Pelcresponder,
Time vs. Information Tradeoffs for Leader Election in Anonymous Trees,
Proc. {\em 27th annual ACM-SIAM symposium on Discrete algorithms}, SODA 2016, 600--609


\bibitem{HK08}
M.A.~Haddar, A.H.~Kacem, Y.~M\'etivier, M.~Mosbah, and M. Jmaiel, 
Electing a Leader in the Local Computation Model using Mobile Agents,
Proc. {\em 6th ACS/IEEE International Conference on Computer Systems and Applications}, AICCSA 2008, 473--480.

\bibitem{H00}
T.~Herman,
Phase clocks for transient fault repairs,
IEEE Transactions on Parallel and Distributed Systems 11(10), 2000,~1048--1057.

\bibitem{HS80}
D.S.~Hirschberg, and J.B.~Sinclair, 
Decentralized Extrema-Finding in Circular Configurations of Processes, 
Communications of the ACM 23, 1980, 627--628.

\bibitem{L77}
G. Le Lann, 
Distributed Systems - Towards a Formal Approach,
Proc. {\em IFIP Congress}, 1977, 155--160.

\bibitem{MA+15}
Y.~Mocquard, E.~Anceaume, J.~Aspnes, Y.~Busnel, and B.~Sericola,
Counting with population protocols,
Proc. {\em IEEE 14th International Symposium on Network Computing and Applications}, NCA 2015, 35--42.

\bibitem{MN+14}
G.B.~Mertzios, S.E.~Nikoletseas, C.~Raptopoulos, and P.G.~Spirakis, 
Determining majority in networks with local interactions and very small local memory. 
Proc. {\em 41st International Colloquium on Automata, Languages, and Programming}, ICALP 2014, 871--882.

\bibitem{MS+14}
O.~Michail and P.G.~Spirakis, 
Simple and efficient local codes for distributed stable network construction, 
Proc. {\em 32nd Annual ACM SIGACT-SIGOPS Symposium on Principles of Distributed Computing}, PODC 2014, 76--85.

\bibitem{P82}
G.L.~Peterson, 
An $O(n log n)$ Unidirectional Distributed Algorithm for the Circular Extrema Problem,
ACM Transactions on Programming Languages and Systems 4, 1982, 758--762.

\bibitem{YK89}
M.~Yamashita and T.~Kameda, 
Electing a Leader when Processor Identity Numbers are not Distinct, 
Proc. {\em 3rd Workshop on Distributed Algorithms}, WDAG 1989, 303--314.

\bibitem{YK96}
M. Yamashita and T. Kameda, 
Computing on Anonymous Networks: Part I - Characterizing the Solvable Cases, 
IEEE Trans. Parallel and Distributed Systems 7, 1996, 69--89.

\end{thebibliography}

\end{document}